\documentclass[a4paper,12pt]{article}

\usepackage{amsmath}
\usepackage{amssymb}
\usepackage{amsthm}
\usepackage{graphicx}
\usepackage[OT4]{fontenc}
\usepackage[utf8]{inputenc}
\usepackage{tikz}
\usepackage{url}
\usepackage[numbers]{natbib}
\usepackage{bm}
\usepackage[margin=3.5cm]{geometry}
\allowdisplaybreaks

\newcommand{\abs}[1]{\left| #1 \right|}

\newcommand{\okra}[1]{\left( #1 \right)}
\newcommand{\mean}[1]{\mathbf{E}\, #1}
\newcommand{\ceil}[1]{\left\lceil #1 \right\rceil}

\newcommand{\kwad}[1]{\left[ #1 \right]}
\newcommand{\klam}[1]{\left\{ #1 \right\}}

\newcommand{\boole}[1]{{\bf 1}{\klam{#1}}}

\DeclareMathOperator{\esssup}{ess\,sup}
\DeclareMathOperator{\PPM}{PPM}
\DeclareMathOperator{\NPD}{NPD}
\DeclareMathOperator{\card}{card}

\DeclareMathOperator{\var}{Var}

\DeclareMathOperator*{\argmax}{arg\,max}

\newtheorem{definition}{Definition}
\newtheorem{example}{Example}
\newtheorem{proposition}{Proposition}
\newtheorem{theorem}{Theorem}
\newtheorem{lemma}{Lemma}

\author{{\L}ukasz D\k{e}bowski%
  \thanks{{\L}. D\k{e}bowski is with the Institute of Computer Science,
    Polish Academy of Sciences, ul. Jana Kazimierza 5, 01-248
    Warszawa, Poland (e-mail: ldebowsk@ipipan.waw.pl). %
  } }

\title{Universal Densities Exist \\ for Every
  Finite Reference Measure} \date{}

\begin{document}

\begin{titlepage}

\maketitle

\begin{abstract}
  As it is known, universal codes, which estimate the entropy rate
  consistently, exist for stationary ergodic sources over finite
  alphabets but not over countably infinite ones. We generalize
  universal coding as the problem of universal densities with respect
  to a fixed reference measure on a countably generated measurable
  space.  We show that universal densities, which estimate the
  differential entropy rate consistently, exist for finite reference
  measures.  Thus finite alphabets are not necessary in some sense.
  To exhibit a universal density, we adapt the non-parametric
  differential (NPD) entropy rate estimator by Feutrill and Roughan.
  Our modification is analogous to Ryabko's modification of prediction
  by partial matching (PPM) by Cleary and Witten.  Whereas Ryabko
  considered a mixture over Markov orders, we consider a mixture over
  quantization levels. Moreover, we demonstrate that any universal
  density induces a strongly consistent Ces\`aro mean estimator of
  conditional density given an infinite past. This yields a universal
  predictor with the $0-1$ loss for a countable alphabet. Finally, we
  specialize universal densities to processes over natural numbers and
  on the real line. We derive sufficient conditions for consistent
  estimation of the entropy rate with respect to infinite reference
  measures in these domains.
  \\[1ex]
  \\[1ex]
\end{abstract}


\end{titlepage}


\section{Introduction}
\label{secIntroduction}

Consider the family of stationary ergodic measures over a given
countable alphabet.  It is known that universal measures, i.e., those
consistently estimating the entropy rate in the almost sure sense and
in expectation, exist for any finite alphabet. A simple example
thereof is the PPM (prediction by partial matching) measure, also
called the $R$-measure, constructed gradually by Cleary and Witten
\cite{ClearyWitten84} and by Ryabko \cite{Ryabko88en2,Ryabko08}.

Universal measures are important for many reasons. They matter not
only in practical data compression but also in various problems of
statistical inference, as advocated in \cite{RyabkoAstolaMalyutov16}.
Here, we name a few examples of their applications:
\begin{itemize}
\item The Shannon-Fano code taken with respect to a universal measure
  is an instance of a lossless universal code for data
  compression. Other important instances of universal codes were
  discovered in
  \cite{ZivLempel77,NeuhoffShields98,KiefferYang00}. These other codes
  do not necessarily induce a universal measure due to the strict
  Kraft inequality.
\item Universal measures are an important building block also in
  estimation of the Markov order \cite{MerhavGutmanZiv89} and of the
  hidden Markov order \cite{ZivMerhav92}. Indirectly, they are also
  connected to upper bounds for mutual information and showing
  disjointness of classes of finite-state and perigraphic processes,
  discussed in statistical language modeling
  \cite{Debowski21b,Debowski21}.
\item As shown in \cite{GyorfiPaliMeulen94}, having a measure that is
  universal in expectation, we can construct a strongly consistent
  Ces\`aro mean estimator of the marginal measure for memoryless
  sources. Since this convergence holds in the Kullback-Leibler
  divergence, it also holds in the total variation---by the Pinsker
  inequality \cite{Pinsker60en,CsiszarKorner11}.
\item Moreover, universal measures induce universal predictors with
  the $0-1$ loss under mild conditions \cite{DebowskiSteifer22}, see
  also \cite{MorvaiWeiss21} for other related results. In particular,
  if there is an estimator of the conditional density given an
  infinite past that is strongly consistent in the total variation, it
  also induces a universal predictor with the $0-1$ loss, see
  \cite{Bailey76,Ornstein78}.
\end{itemize}
It is known, alas, that universal measures or codes do not exist for a
countably infinite alphabet
\cite{Kieffer78,GyorfiPaliMeulen94,SilvaPiantanida20}.  It may seem
that the assumption of a finite alphabet is necessary in general.  In
this paper, we disprove this hypothesis by casting the problem of
universal measures into universal densities, i.e., Radon-Nikodym
derivatives with respect to a given reference measure.

The direct inspiration of the following constructions comes from a
recent paper by Feutrill and Roughan \cite{FeutrillRoughan21}. They
considered a problem of estimating the differential entropy rate
$h_\lambda$ (with respect to the Lebesgue measure $\lambda$) of
Gaussian processes with long memory, such as the fractional Gaussian
noise (FGN) or ARFIMA processes. They observed that the differential
entropy rate can be roughly estimated via the simple non-parametric
differential (NPD) entropy rate estimator, which reads
\begin{align}
  \label{FeutrilRoughan}
  \hat h_{\NPD}(x_1,...,x_n)=
  \hat H\okra{\ceil{\frac{x_1}{\Delta}},...,\ceil{\frac{x_n}{\Delta}}}
  +\log\Delta,
\end{align}
where $\Delta>0$ is a fixed bin width and $\hat H(y_1,...,y_n)$ is a
consistent estimator of the entropy rate for a countably infinite
alphabet. A reasonable choice of $\hat H(y_1,...,y_n)$ can be the
estimator by Kontoyiannis et al.\ \cite{KontoyiannisOthers98}, which
is consistent under the Doeblin condition, whereas its simple
modification discussed in \cite{GaoKontoyiannisBienenstock08} is
consistent for any stationary ergodic process but over a finite
alphabet.  What is interesting, for the estimator by
\cite{KontoyiannisOthers98}, the NPD estimates are empirically quite
close to the true differential entropy rate $h_\lambda$ of FGN and
ARFIMA processes even for $\Delta=1$. Feutrill and Roughan
\cite[Theorem VII.1]{FeutrillRoughan21} tried to argue that the NPD
estimator tends to the differential entropy rate for $\Delta\to 0$ and
$n\to\infty$ but their treatment of the joint limit was not rigorous.

Having learned about this result, we thought that there must be some
solid mathematical idea underneath. Suppose that the observations are
bounded in the unit interval, $x_i\in(0,1]$. Let us take the bin width
$\Delta=2^{-l}$ where $l=0,1,2,...$ and apply a universal measure
$R^l$ for a finite alphabet $\klam{1,2,...,2^l}$ as a special case of
estimator $\hat H(y_1,...,y_n)$. Then we obtain
$\hat h_{\NPD}(x_1,...,x_n)=-\frac{1}{n}\log
R_\lambda^l(x_1,...,x_n)$, where function
\begin{align}
  \label{BoundedDensity}
  R_\lambda^l(x_1,...,x_n):=2^{ln}\cdot
  R^l\okra{\ceil{2^lx_1},...,\ceil{2^lx_n}}
\end{align}
is a probability density. Making one step further, let weights $w_l>0$
with $\sum_{l=0}^\infty w_l=1$. Then we may consider the mixture over all
quantization levels,
\begin{align}
  \label{BoundedMixture}
  R_\lambda(x_1,...,x_n):=\sum_{l=0}^\infty
  w_l\cdot 2^{ln}\cdot
  R^l\okra{\ceil{2^lx_1},...,\ceil{2^lx_n}},
\end{align}
which is also a probability density. The transition from
(\ref{BoundedDensity}) to (\ref{BoundedMixture}) is analogous to
Ryabko's \cite{Ryabko88en2,Ryabko08} derivation of the universal
$R$-measure over a finite alphabet by taking a mixture of PPM measures
by \cite{ClearyWitten84} over all Markov orders. Hence, for any
stationary ergodic process $(X_i)_{i\in\mathbb{Z}}$ with
$X_i\in(0,1]$, we have
\begin{align}
  \limsup_{n\to\infty}\kwad{-\frac{1}{n}\log R_\lambda(X_1,...,X_n)}
  &\le
    \inf_{l\ge 0} h_\lambda^l \text{ a.s.},
\end{align}
where quantized entropy rates $h_\lambda^l$ are the limits of
$-\frac{1}{n}\log R_\lambda^l(X_1,...,X_n)$ by universality of
measures $R^l$. Moreover, since $R_\lambda$ is a probability density
then by the asymptotic equipartition \cite{Barron85} and by
\cite[Theorem 3.1]{Barron85b}, we obtain
\begin{align}
  \liminf_{n\to\infty}\kwad{-\frac{1}{n}\log R_\lambda(X_1,...,X_n)}
  &\ge
    h_\lambda \text{ a.s.}
\end{align}
Thus density $R_\lambda$ estimates entropy rate $h_\lambda$
consistently if $\inf_{l\ge 0} h_\lambda^l = h_\lambda$.

In this paper, we prove that this holds true in a more general
setting. Density (\ref{BoundedMixture}) is an instance of a more
general construction, which we call the NPD density to honor the idea
by Feutrill and Roughan. The NPD density can be defined with respect
to an arbitrary finite reference measure on a countably generated
measurable space. Moreover, if the reference measure is finite then
the respective NPD density is universal, i.e., it estimates
consistently the differential entropy rate with respect to this
reference measure. This result can be proved using the asymptotic
equipartition for densities \cite{Barron85}, Barron's inequality
\cite[Theorem 3.1]{Barron85b}, and the monotone convergence of
$f$-divergences for filtrations \cite[Chapter 22, Problem
4]{Debowski21}, a by-product of our earlier investigations
\cite[Lemma 2]{Debowski09}. The last one yields
$\inf_{l\ge 0} h_\lambda^l = h_\lambda$ indeed.

Thus, it is not finiteness of the alphabet but rather finiteness of
the reference measure that allows for universal densities.
Non-existence of universal codes for a countably infinite alphabet is
tightly connected to non-existence of a uniform probability measure
over the set of natural numbers. The relevance of finite reference
measures may be known to experts in density estimation, see
\cite[Lemma 2]{YangBarron99}, although that particular result concerns
memoryless sources.  In this context, we also recall that
non-existence of universal measures for countably infinite alphabets
rests on non-existence of a consistent density estimator for
memoryless sources over such alphabets
\cite{GyorfiPaliMeulen94}. Thus, there are interactions between
density estimation and universal coding, worth further exploration.

Having constructed the NPD density and proven its universality, we
will also discuss some applications. These are as follows.
\begin{itemize}
\item Inspired by work \cite{GyorfiPaliMeulen94}, we show that if a
  universal density exists (for a given subclass of processes) then it
  induces a strongly consistent Ces\`aro mean estimator of the
  conditional density given an infinite past, in the total variation,
  as it follows by the Pinsker inequality.  Any such conditional
  density estimator solves also the problem of universal prediction
  with the $0-1$ loss for a countable alphabet. This strengthens our
  earlier results from work \cite{DebowskiSteifer22}, which dealt with
  a finite alphabet.

\item As some examples, we specialize NPD densities to the classes of
  processes over a countably infinite alphabet and on the real
  line. We easily name some sufficient conditions that allow for
  consistent estimation of the entropy rate with respect to infinite
  reference measures in these domains, see also
  \cite{SilvaPiantanida20}. In particular, as we show, there exists a
  strongly consistent entropy rate estimator with respect to the
  Lebesgue measure in the class of stationary ergodic Gaussian
  processes.
\end{itemize}

The organization of the paper is as follows.  In Section
\ref{secPreliminaries}, we amply recall preliminaries to keep the
paper relatively self-contained.  Section \ref{secMain} contains the
main result, i.e., the construction of the universal NPD density with
respect to a finite reference measure.  Having an instance of a
universal density, in Section \ref{secApplications}, we construct a
consistent Ces\`aro mean estimator of the conditional density given an
infinite past.  Consequently, this estimator of conditional density is
turned into a universal predictor in the countable alphabet setting.
As some further examples, in Section \ref{secExamples}, we present
applications of the total NPD density to processes over a countably
infinite alphabet and on the real line. The paper is concluded in
Section \ref{secConclusion}, where we sketch open problems.

\section{Preliminaries}
\label{secPreliminaries}

In this section we establish our setting and report the received
knowledge.  We discuss the introductory material such as the entropy
rate, the asymptotic equipartition, the definition of universal
measures, and the universal PPM measure for a finite alphabet.

\subsection{General setting}
\label{secSetting}

Let $(\mathbb{X},\mathcal{X},\mu)$ be a countably generated measurable
space with a $\sigma$-finite measure $\mu$ on it.  Measure $\mu$ will
be called the reference measure. Simple familiar examples are the
counting measure $\mu(A)=\gamma(A):=\card A$ for a countable alphabet
$\mathbb{X}$ or the Lebesgue measure $\mu([a,b])=\lambda([a,b]):=b-a$
for $\mathbb{X}=\mathbb{R}$.  Consider the product space
$(\mathbb{X}^{\mathbb{Z}},\mathcal{X}^{\mathbb{Z}})$ and put random
variables
$X_k:\mathbb{X}^{\mathbb{Z}}\ni(x_i)_{i\in\mathbb{Z}}\mapsto
x_k\in\mathbb{X}$.  We write the tuples of points as
$x_{j:k}:=(x_j,x_{j+1},...,x_k)$.  For a probability measure $R$ on
$(\mathbb{X}^{\mathbb{Z}},\mathcal{X}^{\mathbb{Z}})$, we denote its
finite-dimensional restrictions $R_n(A):=R(X_{1:n}\in A)$ and if
$R_n\ll \mu^n$ then we write the densities
\begin{align}
 R_\mu(x_{1:n}):=\frac{dR_n}{d\mu^n}(x_{1:n}). 
\end{align}
The space of stationary ergodic measures on
$(\mathbb{X}^{\mathbb{Z}},\mathcal{X}^{\mathbb{Z}})$ with respect to
the shift operation will be denoted as $\mathbb{E}$. A measure
$P\in\mathbb{E}$ where $P_n\ll \mu^n$ is called a memoryless source if
$P_\mu(x_{1:n})=\prod_{i=1}^n P_\mu(x_i)$.

Fix a measure $P\in\mathbb{E}$ where $P_n\ll \mu^n$.  Convergence in
probability with respect to measure $P$ is denoted
\begin{align}
  \lim_{n\to\infty} Y_n=Y \text{ i.p.}\iff \forall_{\epsilon>0}
  \lim_{n\to\infty} P(|Y_n-Y|>\epsilon)=0.
\end{align}
The expectation operator $\mean$ and the quantifier ``almost surely''
(a.s.)\ will be also taken throughout with respect to $P$.  Throughout
the paper, symbol $\log$ denotes the logarithm to a fixed
underspecified base. We define the block entropy
\begin{align}
  h_\mu(n):=\mean\kwad{-\log P_\mu(X_{1:n})}=
  -\int P_\mu(x_{1:n})\log P_\mu(x_{1:n}) d\mu^n(x_{1:n}).
\end{align}
A short notice, if the reference measure $\mu$ is the counting measure
then $h_\mu(n)\ge 0$ since $P_\mu(x_{1:n})\le 1$. By contrast, if
$\mu$ is a probability measure then $h_\mu(n)\le 0$ since $-h_\mu(n)$
is the Kullback-Leibler divergence between $P_n$ and $\mu^n$.

In any case, by stationarity and by the Jensen inequality, the block
entropy is subadditive, $h_\mu(n+m)\le h_\mu(n)+h_\mu(m)$. Hence by
the Fekete lemma for subadditive sequences \cite{Fekete23}, see also
\cite[Theorem 5.11]{Debowski21}, sequence $h_\mu(n)/n$ is decreasing
and there exists the entropy rate
\begin{align}
  h_\mu:=\lim_{n\to\infty} \frac{h_\mu(n)}{n}
  =\inf_{n\ge 1} \frac{h_\mu(n)}{n}.
\end{align}
Moreover, by the Shannon-McMillan-Breiman (SMB) theorem, in many
cases, we have the asymptotic equipartition property, namely,
\begin{align}
  \label{SMB}
  \lim_{n\to\infty} \kwad{-\log P_\mu(X_{1:n})}/n=h_\mu \text{ a.s.},
\end{align}
as noticed by Shannon \cite{Shannon48} for memoryless sources and
consecutively generalized in
\cite{Breiman57,Chung61,Barron85,AlgoetCover88} to other settings.

In particular, let us consider the class of stationary ergodic
measures with a finite entropy rate,
\begin{align}
  \mathbb{E}(\mu):=\klam{P\in\mathbb{E}:
  P_n\ll \mu^n \text{ and }
  \abs{h_\mu}<\infty}.
\end{align}
Let us also denote the conditional density
\begin{align}
  P_\mu(x_{n+1:n+m}|x_{1:n}):=\frac{P_\mu(x_{1:n+m})}{P_\mu(x_{1:n})}.
\end{align}
As shown by Barron \cite[Theorem 1]{Barron85}, we have the asymptotic
equipartition (\ref{SMB}) for an arbitrary $P\in\mathbb{E}(\mu)$. This
fact is a consequence of the Breiman ergodic theorem \cite{Breiman57},
\cite[Theorem 12(c)]{Algoet94} since, as shown by Barron \cite[Proof
of Theorem 1]{Barron85}, for $P\in\mathbb{E}(\mu)$ there exists the
limit of conditional densities
\begin{align}
  P_\mu(X_0|X_{-\infty:-1}):=\lim_{n\to\infty}
  P_\mu(X_0|X_{-n:-1}) \text{ a.s.}
\end{align}
and $\mean\sup_{n\in\mathbb{N}}\abs{\log
  P_\mu(X_0|X_{-n:-1})}<\infty$, whereas
\begin{align}
  h_\mu=\mean \kwad{-\log P_\mu(X_0|X_{-\infty:-1})}.
\end{align}

Obviously, the asymptotic equipartition is an interesting property
from a statistical perspective since it suggests that the entropy rate
can be estimated, which is known to be the case for a finite counting
reference measure. However, the problem of applying statement
(\ref{SMB}) for estimating the entropy rate is that we need to
estimate the unknown probability measure.

\subsection{Universal measures}
\label{secUniversal}

To approach the problem of entropy estimation in a reasonable way, let
us consider another probability measure $R$ where $R_n\ll \mu^n$,
which need not be stationary or ergodic.  As it is a part of an older
information-theoretic folklore, for any $m>1$ we have a particular
case of Markov's inequality, called sometimes the Barron inequality,
\begin{align}
  \label{Barron}
  P\okra{\log \frac{dP_n}{dR_n}(X_{1:n})\le -\log m}&\le \frac{1}{m},
\end{align}
shown for example in \cite[Theorem 3.1]{Barron85b} or applied
implicitly much earlier in \cite{Chaitin75}.  Since
$dP_n/dR_n(x_{1:n})=P_\mu(x_{1:n})/R_\mu(x_{1:n})$, as a result, by an
easy application of the Borel-Cantelli lemma, we obtain
\begin{align}
  \label{SourceAS}
  \liminf_{n\to\infty} \kwad{-\log R_\mu(X_{1:n})}/n&\ge h_\mu \text{ a.s.}
\end{align}
Since the Kullback-Leibler divergence is non-negative in general,
\begin{align}
  D(P_n||R_n)&:=\mean\kwad{\log \frac{dP_n}{dR_n}(X_{1:n})}
               =\mean\kwad{\log \frac{P_\mu(X_{1:n})}{R_\mu(X_{1:n})}}
               \ge 0,
\end{align}
we also have a similar result in expectation:
\begin{align}
  \label{SourceExp}
  \liminf_{n\to\infty} \mean\kwad{-\log R_\mu(X_{1:n})}/n&\ge h_\mu \text{ a.s.}
\end{align}

Consequently, we will consider three definitions of a universal
measure. By these definitions, a universal measure can be used for
estimating the entropy rate of an unknown stationary ergodic process.
\begin{definition}
  \label{defiUniAS}
  A probability measure $R$ where $R_n\ll \mu^n$ is called
  strongly universal with respect to a measure $\mu$
  if for every measure $P\in\mathbb{E}(\mu)$,
  \begin{align}
    \label{UniAS}
    \lim_{n\to\infty} \kwad{-\log R_\mu(X_{1:n})}/n&= h_\mu \text{ a.s.}
  \end{align}  
\end{definition}
\begin{definition}
  \label{defiUniExp}
  A probability measure $R$ where $R_n\ll \mu^n$ is called
  universal in expectation with respect to a measure
  $\mu$ if for every measure $P\in\mathbb{E}(\mu)$,
  \begin{align}
    \label{UniExp}
    \lim_{n\to\infty} \mean\kwad{-\log R_\mu(X_{1:n})}/n&= h_\mu.
  \end{align}  
\end{definition}
\begin{definition}
  \label{defiUniIP}
  A probability measure $R$ where $R_n\ll \mu^n$ is called universal
  in probability with respect to a measure $\mu$ if for every measure
  $P\in\mathbb{E}(\mu)$,
  \begin{align}
    \label{UniIP}
    \lim_{n\to\infty} \kwad{-\log R_\mu(X_{1:n})}/n&= h_\mu \text{ i.p.}
  \end{align}  
\end{definition}
\noindent
\emph{Remark:} A measure $R$ that satisfies (\ref{UniAS}),
(\ref{UniExp}), and (\ref{UniIP}) is simply called universal with
respect to $\mu$. If the reference measure is the counting measure
over a countable alphabet $\mathbb{X}$, i.e.,
$\mu(A)=\gamma(A):=\card A$ for $A\subset\mathbb{X}$, then we speak of
measures that are universal with respect to alphabet $\mathbb{X}$,
respectively.  In this case, we drop the subscript $\gamma$:
$P_\gamma(x)\to P(x)$, $R_\gamma(x)\to R(x)$, and $h_\gamma\to h$.

In fact, implications (\ref{UniAS}) $\implies$ (\ref{UniExp})
$\implies$ (\ref{UniIP}) hold usually. Obviously, (\ref{UniAS})
$\implies$ (\ref{UniIP}) since the almost sure convergence implies the
convergence in probability by the Riesz theorem. By contrast,
implications (\ref{UniAS}) $\implies$ (\ref{UniExp}) and
(\ref{UniExp}) $\implies$ (\ref{UniIP}) need a more explicit
information-theoretic proof. In the next two propositions we
generalize facts that are well known in the case of the counting
reference measure.

\begin{proposition}
  \label{theoUniExp}
  If measure $R$ is strongly universal with respect to measure $\mu$
  and $R_\mu(x_{1:n})\ge Kc^n$ holds uniformly for some $0<K<1$ and
  $c>0$ then measure $R$ is universal in expectation with respect to
  measure $\mu$.
\end{proposition}
\begin{proof}
  Since $R_\mu(x_{1:n})\ge Kc^n$, we have
  $\kwad{-\log R_\mu(x_{1:n})}/n\le -\log c-\log K$.
  Hence by the Fatou lemma with the minus sign and by the strong
  universality of measure $R$, we obtain
  \begin{align}
    \limsup_{n\to\infty} \frac{\mean\kwad{-\log R_\mu(X_{1:n})}}{n}
    \le
    \mean
    \limsup_{n\to\infty} \frac{\kwad{-\log R_\mu(X_{1:n})}}{n}
    =
    h_\mu.
  \end{align}
  Combining this with inequality (\ref{SourceExp}) yields the claim.
\end{proof}

A careful reader may have noticed that the above proof applies the
Fatou lemma rather than the dominated convergence. The latter is
usually invoked for proving universality with respect to the counting
reference measure.

\begin{proposition}
  \label{theoUniIP}
  A measure universal in expectation is universal in probability.
\end{proposition}
\begin{proof}
  Let measure $R$ be universal in expectation. Let an $\epsilon>0$. In
  view of the asymptotic equipartition (\ref{SMB}) and inequality
  (\ref{SourceAS}), it is sufficient
  to show 
  \begin{align}
    \lim_{n\to\infty} P\okra{-\frac{\log R_\mu(X_{1:n})}{n}
    +\frac{\log P_\mu(X_{1:n})}{n}\ge \epsilon}=0.
  \end{align}
  Hence, by the Markov inequality,
  $P(X\ge\epsilon)\le\mean X_+/\epsilon$, it suffices to prove
  \begin{align}
    \lim_{n\to\infty}
    \frac{1}{n}\mean\kwad{\log \frac{P_\mu(X_{1:n})}{R_\mu(X_{1:n})}}_+=0.
  \end{align}
  But this follows from equality
  \begin{align}
    \frac{1}{n}\mean\kwad{\log
    \frac{P_\mu(X_{1:n})}{R_\mu(X_{1:n})}}_+
    =
    \frac{1}{n}\mean\kwad{\log
    \frac{P_\mu(X_{1:n})}{R_\mu(X_{1:n})}}
    +
    \frac{1}{n}\mean\kwad{\log
    \frac{P_\mu(X_{1:n})}{R_\mu(X_{1:n})}}_-.
  \end{align}
  The first term on the right hand side tends to $0$ by the
  universality in expectation whereas the second term tends to $0$
  since
  \begin{align}
    \mean\kwad{\log
    \frac{P_\mu(X_{1:n})}{R_\mu(X_{1:n})}}_-
    &=
    \int_0^\infty P\okra{\log
      \frac{P_\mu(X_{1:n})}{R_\mu(X_{1:n})}\le -u}du
      \nonumber\\
    &\le
    \int_0^\infty \frac{d(\log m)}{m}
    =\int_1^\infty \frac{dm}{m^2}\log e=\log e
  \end{align}
  is uniformly bounded by the Barron inequality (\ref{Barron}).
\end{proof}

Thus it is sufficient usually to demonstrate strong universality and a
uniform lower bound on a given density. Other flavors of universality
follow hence automatically.

\subsection{Finite alphabet}

For a finite alphabet, universal measures are known to exist.  We
recall an example of a universal measure for this case, called the PPM
(prediction by partial matching). The PPM measure was discovered
gradually.  Cleary and Witten \cite{ClearyWitten84} considered roughly
Markov approximations $\PPM_k$ defined in (\ref{PPMk}) and coined name
PPM, which we prefer as more distinctive. Ryabko
\cite{Ryabko88en2,Ryabko08} considered the infinite series $\PPM$
defined in (\ref{PPM}), called it the $R$-measure, and proved that it
is universal. To be precise, Ryabko used the Krichevsky-Trofimov
smoothing ($+1/2$) rather than the Laplace smoothing ($+1$) applied in
(\ref{PPMk}). This difference is minor and does not affect
universality.

The definition that we use is as follows.
\begin{definition}[PPM density]
  Let the alphabet be $\mathbb{X}=\klam{a_1,...,a_D}$. Adapting the
  definitions from works
  \cite{ClearyWitten84,Ryabko88en2,Ryabko08,Debowski18}, the PPM
  density of order $k\ge 0$ is defined as
  \begin{align}
   \label{PPMk}
    \PPM_k^D(x_{1:n})
    &:=
      \begin{cases}
        \displaystyle
      D^{-k-1}\prod_{i=k+2}^n
      \frac{N(x_{i-k:i}|x_{1:i-1})+1}{N(x_{i-k:i-1}|x_{1:i-2})+D},
      & k\le n-2,
      \\
      D^{-n}, & k\ge n-1,
      \end{cases}
  \end{align}
  where the frequency of a substring $w_{1:k}$ in a string $x_{1:n}$ is
\begin{align}
    N(w_{1:k}|x_{1:n}):=\sum_{i=1}^{n-k+1}\boole{x_{i:i+k-1}=w_{1:k}}.
\end{align}
Subsequently, we define the (total) PPM density as
\begin{align}
  \label{PPM}
  \PPM^D(x_{1:n})
  &:=
  \sum_{k=0}^\infty
    w_k\PPM_k^D(x_{1:n})
    ,
  \quad
  w_k:=\frac{1}{k+1}-\frac{1}{k+2}.
\end{align}  
\end{definition}

Infinite series (\ref{PPM}) is a sum over finitely many distinct terms
and it is computable in the sense of computability theory, see
\cite{DebowskiSteifer22}, since we have
\begin{align}
  \label{PPMbound}
  \PPM_k^D(x_{1:n})=D^{-n} \text{ for } k\ge L(x_{1:n}),
\end{align}
where $L(x_{1:n})$ is the maximal length of a repetition in $x_{1:n}$,
\begin{align}
  L(x_{1:n}):=\max\klam{k\ge 0: x_{i+1:i+k}=x_{j+1:j+k} \text{ for
  some }
  0\le i<j\le n-k},
\end{align}
an important information-theoretic statistic in its own right
\cite[Chapter 9]{Debowski21}.

Let us show that the total PPM density yields a probability measure.
\begin{theorem}
  There exists a measure $R$ such that $R(x_{1:n})=\PPM^D(x_{1:n})$.
\end{theorem}
\noindent
\emph{Remark:} This measure $R$ will be denoted $\PPM^D$.
\begin{proof}
  By the Kolmogorov process theorem, it suffices to show that
  \begin{align}
    \sum_{x_{n+1}\in\mathbb{X}} \PPM^D(x_{1:n+1})=\PPM^D(x_{1:n}).
  \end{align}
  But this follows by the monotone convergence from
  \begin{align}
    \sum_{x_{n+1}\in\mathbb{X}} \PPM_k^D(x_{1:n+1})=\PPM_k^D(x_{1:n}),
  \end{align}
  which in turn follows by the definition of $\PPM_k^D(x_{1:n})$.
\end{proof}

The strong universality of the total PPM measure follows by the
Stirling approximation and by the Birkhoff ergodic theorem. Moreover,
we have the uniform lower bound
\begin{align}
  \PPM^D(x_{1:n})\ge w_{n-1}\PPM_{n-1}^D(x_{1:n})
  \ge \frac{D^{-n}}{(n+1)^2}>\frac{1}{4}\,(2D)^{-n}.
\end{align}
Hence by Propositions \ref{theoUniExp} and \ref{theoUniIP} follows
also the universality in expectation and in probability.
\begin{theorem}[\cite{Ryabko88en2,Ryabko08,Debowski18}]
  Measure $\PPM^D$ is universal with respect to alphabet
  $\mathbb{X}=\klam{a_1,...,a_D}$.
\end{theorem}

\section{Main results}
\label{secMain}

In this section, we exhibit the main result. We show that universal
measures exist if the reference measure is an arbitrary finite
measure. We do it by an effective construction. Our constructive
example of a universal measure is called the NPD (non-parametric
differential) measure to honor the quantization idea by Feutrill and
Roughan \cite{FeutrillRoughan21}. They carried out the construction of
the NPD estimator half-way---as detailed in Section
\ref{secIntroduction}.

Let us proceed to the construction of the NPD measure.  Let
notation $\mathcal{X}_l\uparrow\mathcal{X}$ denote a filtration of a
$\sigma$-field $\mathcal{X}$, i.e., a sequence of nested
$\sigma$-fields $\mathcal{X}_l\subset \mathcal{X}_{l+1}$ where
$l=0,1,2,...$ and
$\sigma(\bigcup_{l\ge 0} \mathcal{X}_l)=\mathcal{X}$.  Assuming that
the reference measure $\mu$ is a finite measure on a countably
generated measurable space, we will demonstrate universality of the
following constructive object.
\begin{definition}[NPD density]
  Let $(\mathbb{X},\mathcal{X},\mu)$ be a countably generated finite
  measure space.  Let $\mathcal{X}_l\uparrow\mathcal{X}$ where
  $l=0,1,2,...$ be a filtration where the $\sigma$-fields
  $\mathcal{X}_l$ are finite with
  $\mathcal{X}_0=\klam{\mathbb{X},\emptyset}$. Such a filtration
  exists since $\mathcal{X}$ is countably generated. Let $\chi_l$ be
  the finite partitions that generate $\sigma$-fields $\mathcal{X}_l$
  respectively.  We treat classes $\chi_l$ as finite alphabets of
  symbols $A\in \chi_l$.  We introduce quantizations of points
  $x\in\mathbb{X}$ as symbols $x^l:=A$ for $x\in A\in
  \chi_l$. Moreover, for $l=0,1,2,...$, let $R^l$ be certain measures
  that are universal for alphabets $\chi_l$.  We define the NPD
  density of order $l\ge 0$ as
  \begin{align}
  \label{NPDml}
    \NPD^l_\mu(x_{1:n})&:=\frac{R^l(x_{1:n}^l)}{\prod_{i=1}^n\mu(x_i^l)}.
  \end{align}
  Subsequently, we define the (total) NPD density as
  \begin{align}
    \label{NPDm}
    \NPD_\mu(x_{1:n})&:=\sum_{l=0}^\infty w_l \NPD^l_\mu(x_{1:n}),
    \quad
    w_l:=\frac{1}{l+1}-\frac{1}{l+2}.
  \end{align}
\end{definition}

Let us note that the NPD measure depends implicitly on filtration
$\mathcal{X}_l\uparrow\mathcal{X}$, universal measures $R^l$ for
finite alphabets, and reference measure $\mu$.  Actually, for the
universality of the NPD density, it does not matter which universal
measures $R^l$ we use in definition (\ref{NPDml}).  There is some
analogy between the PPM series (\ref{PPM}) and the NPD series
(\ref{NPDm}).  Weights $w_k$ in series (\ref{PPM}) weigh different
Markov approximations, whereas weights $w_l$ in series (\ref{NPDm})
weigh different quantization levels.  Thus, we may say that our
development of the quantization idea by Feutrill and Roughan
\cite{FeutrillRoughan21} is analogical to Ryabko's
\cite{Ryabko88en2,Ryabko08} development of the PPM measures by Cleary
and Witten \cite{ClearyWitten84}. Whereas Cleary and Witten
\cite{ClearyWitten84} and Feutrill and Roughan
\cite{FeutrillRoughan21} considered approximations of a fixed order,
the order meaning the Markov order or the quantization level
respectively, the idea of Ryabko \cite{Ryabko88en2,Ryabko08} and of us
is to apply a mixture of infinitely many orders. As we have seen, this
guarantees that the total PPM measure is universal and it is
reasonable to expect that so is the total NPD density.

Before we demonstrate universality, let us take a closer look at
the NPD densities.  The total NPD density is measurable and finite
$\mu$-almost everywhere, as it follows by the monotone
convergence. Just an explicit proof for a sanity check.
\begin{theorem}
  We have $\NPD_\mu(x_{1:n})<\infty$ for $\mu^n$-almost all $x_{1:n}$.
\end{theorem}
\begin{proof}
  For each $l\ge 0$, we have
  \begin{align}
    \int \NPD^l_\mu(x_{1:n}) d\mu^n(x_{1:n})=\sum_{x_{1:n}^l}R^l(x_{1:n}^l)=1.
  \end{align}
  Since $\NPD^l_\mu(x_{1:n})\ge 0$, hence by the monotone convergence,
  we obtain
  \begin{align}
    \int \NPD_\mu(x_{1:n}) d\mu^n(x_{1:n})=
    \sum_{l=0}^\infty w_l \int \NPD^l_\mu(x_{1:n}) d\mu^n(x_{1:n})=1.
  \end{align}
  Since the integral is finite, the integrand is finite almost
  everywhere.
\end{proof}

Although the total NPD density can be divergent for particular tuples
$x_{1:n}$, we can control its finiteness pretty well in some important
cases.
\begin{example}[dyadic partitions]
  \label{exQuantiles}
  Let the universal measures in (\ref{NPDml}) be the PPM measures,
  $R^l=\PPM^{\card \chi_l}$.  Then
\begin{align}
  \NPD^l_\mu(x_{1:n})
  &=\frac{(\card \chi_l)^{-n}}{\prod_{i=1}^n\mu(x_i^l)}
    \text{ for } l\ge M(x_{1:n}),
\end{align}
where $M(x_{1:n})$ is the minimal quantization level that puts points
$x_1,...,x_n$ into different bins,
\begin{align}
 M(x_{1:n}):=\min\klam{l\ge 0: L(x_{1:n}^l)=0}.
\end{align}
A particularly regular case arises for uniformly dyadic partitions:
\begin{align}
  \label{Dyadic}
  \card \chi_l=2^l\text{ and }\mu(x_i^l)=2^{-l}.
\end{align}
Such partitioning is feasible if the reference measure $\mu$ is a
non-atomic probability measure, such as the normal distribution
$N(m,\sigma^2)$ to be discussed in Section \ref{secReal}. Then
$\NPD^l_\mu(x_{1:n})=1$ for $l\ge M(x_{1:n})$ and series
$\NPD_\mu(x_{1:n})$ is finite if $M(x_{1:n})$ is finite, whereas
statistic $M(X_{1:n})$ is finite $\mu^{\mathbb{Z}}$-almost surely.
\end{example}

Universality was stated in Definitions
\ref{defiUniAS}--\ref{defiUniIP} as a property of measures rather than
their densities. Thus, let us see the following statement.
\begin{theorem}
  There exists a measure $R$ such that
  $R_\mu(x_{1:n})=\NPD_\mu(x_{1:n})$.
\end{theorem}
\noindent
\emph{Remark:} This measure $R$ will be denoted $\NPD$.
\begin{proof}
  We may construct measures
  \begin{align}
    R_n(A):=\int_{x_{1:n}\in A} \NPD_\mu(x_{1:n})d\mu^n(x_{1:n}).
  \end{align}
  To show that measures $R_n$ induce measure $R$ on infinite
  sequences, by the Kolmogorov process theorem, it suffices to show
  that $R_{n+1}(A\times\mathbb{X})=R_n(A)$. In turn, using the Fubini
  theorem this is implied by condition
  \begin{align}    
    \int_{x_{n+1}\in\mathbb{X}}
    \NPD_\mu(x_{1:n+1})d\mu(x_{n+1})=\NPD_\mu(x_{1:n}), 
  \end{align}
  The above follows by the monotone convergence from
  \begin{align}
    \int_{x_{n+1}\in\mathbb{X}}
    \NPD^l_\mu(x_{1:n+1})d\mu(x_{n+1})=\NPD^l_\mu(x_{1:n}),
  \end{align}
  which is true since each $R^l$ in (\ref{NPDml}) is a measure.
\end{proof}

In order to prove universality of the total NPD measure, we
will apply a lemma that concerns convergence of $f$-divergences for
filtrations:
\begin{lemma}[\mbox{\cite[Chapter 3, Problem 4]{Debowski21}}]
  \label{theoConvexMartingale} For an interval $A$, let
  $f:A\rightarrow[0,\infty]$ be a non-negative, continuous, and convex
  measurable function, let $\nu\ll \rho$ be two finite measures on a
  measurable space, and let $\mathcal{G}_n\uparrow \mathcal{G}$ be a
  filtration. We have
  \begin{align}
    \label{ConvexMartingale}
    \lim_{n\rightarrow\infty}
    \int f\okra{\frac{d\nu|_{\mathcal{G}_n}}{d\rho|_{\mathcal{G}_n}}}d\rho
    =
    \int f\okra{\frac{d\nu|_{\mathcal{G}}}{d\rho|_{\mathcal{G}}}}d\rho
    ,
  \end{align}
  where the sequence on the left hand side is increasing.
\end{lemma}
\noindent 
\emph{Remark:} Lemma \ref{theoConvexMartingale} follows, via the
martingale convergence, by a synergy of the Fatou lemma and the Jensen
inequality. The Fatou lemma yields that the left hand side is larger
than the right hand side, whereas the Jensen inequality yields the
reversed inequality. The idea is the same as the proof of \cite[Lemma
2]{Debowski09}, which concerns continuity of conditional mutual
information for $\sigma$-fields.

Now we will derive the main result of this section. 
\begin{theorem}
  \label{theoFiniteAS}
  Measure $\NPD$ is universal with respect to the finite measure
  $\mu$.
\end{theorem}
\begin{proof}
  It suffices to show that for $P\in\mathbb{E}(\mu)$, we
  have
  \begin{align}
    \label{SufficientAS}
    \limsup_{n\to\infty} \frac{\kwad{-\log \NPD_\mu(X_{1:n})}}{n}
    \le
    h_\mu \text{ a.s.}
  \end{align}
  The strong universality follows hence by the converse bound
  (\ref{SourceAS}). By contrast, the universality in expectation and
  in probability follows by Propositions
  \ref{theoUniExp} and \ref{theoUniIP} and inequality
  \begin{align}
    \NPD_\mu(x_{1:n})\ge\frac{1}{2}\NPD^0_\mu(x_{1:n})
    =\frac{1}{2}\,\mu(\mathbb{X})^{-n}.
  \end{align}

  So as to demonstrate (\ref{SufficientAS}), we first observe that
  by the strong universality of measures $R^l$, we have
  \begin{align}
    \label{PPMAS}
    \lim_{n\to\infty} \frac{\kwad{-\log R^l(X_{1:n}^l)}}{n}
    &=
    \inf_{n\ge 1} \frac{\mean\kwad{-\log
    P_n(X_{1:n}^l)}}{n}
      \text{ a.s.}
  \end{align}
  On the other hand, by the Birkhoff ergodic theorem, we have
  \begin{align}
    \label{MUAS}
    \lim_{n\to\infty}\frac{1}{n}\sum_{i=1}^n\kwad{-\log\mu(X_i^l)}
    = C_l:=
    \mean \kwad{-\log\mu(X_i^l)} \text{ a.s.}
  \end{align}
  Moreover, each cross entropy $C_l$ is finite since each $C_l$ is a
  sum over finitely many finite elements---by $P_1\ll\mu$.  Denote
  quantities
  \begin{align}
    h^l_\mu(n):=
    \mean\kwad{-\log \frac{P_n(X_{1:n}^l)}{\prod_{i=1}^n\mu(X_i^l)}}
    =\mean\kwad{-\log P_n(X_{1:n}^l)}-n\,C_l.
  \end{align}  
  Since cross entropies $C_l$ are finite, equations (\ref{PPMAS}) and
  (\ref{MUAS}) imply
  \begin{align}
    \lim_{n\to\infty} \frac{\kwad{-\log \NPD^l_\mu(X_{1:n})}}{n}
    &= h^l_\mu:= \inf_{n\ge 1} \frac{h^l_\mu(n)}{n}
    \text{ a.s.}
  \end{align}
  Since $\NPD_\mu(x_{1:n})\ge w_l\NPD^l_\mu(x_{1:n})$ then for
  any $l\ge 0$, we obtain
  \begin{align}
    \limsup_{n\to\infty} \frac{\kwad{-\log \NPD_\mu(X_{1:n})}}{n}
    &\le
      \limsup_{n\to\infty} \frac{\kwad{-\log w_l\NPD^l_\mu(X_{1:n})}}{n}
      = h^l_\mu
    \text{ a.s.}
  \end{align}
  
  It remains to show that $\inf_{l\ge 0} h^l_\mu=h_\mu$. For this goal
  we observe that
  \begin{align}
    \frac{P_n(x_{1:n}^l)}{\prod_{i=1}^n\mu(x_i^l)}=
    \frac{dP_n|_{\mathcal{X}_l^n}}{d\mu^n|_{\mathcal{X}_l^n}}(x_{1:n}).   
  \end{align}
  Hence we have
  \begin{align}
    h^l_\mu=
    \inf_{n\ge 1} \frac{h^l_\mu(n)}{n}
    =
    \inf_{n\ge 1} \frac{1}{n}\int
    \eta\okra{\frac{dP_n|_{\mathcal{X}_l^n}}{d\mu^n|_{\mathcal{X}_l^n}}}d\mu^n,
  \end{align}
  where $\eta(x):=-x\log x$. We switch the order of infimums,
  \begin{align}
    \inf_{l\ge 0} h^l_\mu
    &=
      \inf_{n\ge 1} \frac{1}{n} \inf_{l\ge 0} \int
      \eta\okra{\frac{dP_n|_{\mathcal{X}_l^n}}{d\mu^n|_{\mathcal{X}_l^n}}}d\mu^n
  \end{align}
  and we apply Lemma \ref{theoConvexMartingale} to function
  $f(x)=\log 2-\eta(x)$.  Hence
  \begin{align}
    \inf_{l\ge 0} h^l_\mu
    &=
      \inf_{n\ge 1} \frac{1}{n} \int
      \eta\okra{\frac{dP_n}{d\mu^n}}d\mu^n
      =
      \inf_{n\ge 1} \frac{h_\mu(n)}{n}
      = h_\mu.
  \end{align}
  The proof is complete.
\end{proof}

\section{Applications}
\label{secApplications}

In this section, we show that if a universal measure exists then we
may construct a strongly consistent Ces\`aro mean estimator for the
limiting conditional density given an infinite past. Subsequently, we
show that any strongly consistent estimator of this conditional
density yields a universal predictor for a countable alphabet.

\subsection{Conditional density estimation}
\label{secDensity}

Here, we will show that a universal density, if it exists, induces a
strongly consistent Ces\`aro mean estimator of the conditional
density. For this goal, we consider a general reference measure $\mu$
as in Section \ref{secSetting}.

A stochastic process $(R_\mu^{(n)})_{n\ge 1}$ is called a conditional
density estimator if each random variable $R_\mu^{(n)}(x)$ is a
(marginal) probability density with respect to measure $\mu$ and each
function $R_\mu^{(n)}$ is a measurable function of random variables
$X_{-n:-1}$. We also denote the respective random measure
$R^{(n)}(A):=\int_A R_\mu^{(n)}(x) d\mu(x)$. We will seek for a
conditional density estimator that for any, $P\in\mathbb{E}(\mu)$,
converges in some strong sense to the conditional density
\begin{align}
  P_\mu^{(\infty)}(x):=\lim_{n\to\infty} P_\mu^{(n)}(x),
  \text{ where }
  P_\mu^{(n)}(x):= P_\mu(x|X_{-n:-1}).
\end{align}
We denote the respective random measure
$P^{(\infty)}(A):=\int_A P_\mu^{(\infty)}(x) d\mu(x)$. Obviously
$P^{(\infty)}=P_1$ if $P$ is a memoryless source.

Inspired by the construction of \cite{GyorfiPaliMeulen94}, we will
consider the following object:
\begin{definition}
  Consider a probability measure $R$ where $R_n\ll\mu^n$.  The Ces\`aro mean
  measure $\bar R$ is defined via conditional densities
  \begin{align}
    \bar R_\mu(x_n|x_{1:n-1}):=
    \frac{1}{n}\sum_{i=0}^{n-1} R_\mu(x_n|x_{n-i:n-1}).
  \end{align}
\end{definition}
\noindent
Let us observe that we may introduce a conditional density estimator
\begin{align}
  \bar R_\mu^{(n)}(x):=\bar R_\mu(x|X_{-n:-1}).
\end{align}
We will call it the Ces\`aro mean density estimator.  We denote the
respective random measure
$\bar R^{(n)}(A):=\int_A \bar R_\mu^{(n)}(x) d\mu(x)$.  In work
\cite{GyorfiPaliMeulen94}, a similar conditional density estimator was
considered, albeit with a reflected time arrow.  The Ces\`aro mean
measure $\bar R$ reminds also of linear interpolation models used for
statistical language modeling in the 1990's \cite{Jelinek97}.

For the Ces\`aro mean density estimator, we have the following result which
generalizes \cite[Theorem 1]{GyorfiPaliMeulen94} for memoryless
sources.
\begin{theorem}
  \label{theoInduced}
  Consider a measure $P\in\mathbb{E}(\mu)$, a probability measure $R$
  where $R_n\ll\mu^n$, and the Ces\`aro mean density estimator
  $(\bar R_\mu^{(n)})_{n\ge 1}$.  We have
  \begin{align}
    \mean\kwad{-\log R_\mu(X_{1:n})}/n-h_\mu
    \ge
    \mean D(P^{(\infty)}||\bar R^{(n)}).
  \end{align}
\end{theorem}
\begin{proof}
  We essentially apply the proof idea of \cite[Theorem
  1]{GyorfiPaliMeulen94}, which is a restriction of the present claim
  to $\mu$ being the counting measure and $P$ being a memoryless
  source. In the reasoning rewritten in a more transparent notation we
  apply the stationarity and the Jensen inequality,
  \begin{align}
    \mean\kwad{-\log R_\mu(X_{1:n})}/n-h_\mu
    &=\mean\kwad{-\frac{1}{n}\sum_{i=1}^n
      \log R_\mu(X_i|X_{1:i-1})}-h_\mu
      \nonumber\\
    &=
      \mean\kwad{-\frac{1}{n}\sum_{i=0}^{n-1}
      \log R_\mu(X_0|X_{-i:-1})}-h_\mu
      \nonumber\\
    &\ge
    \mean\kwad{-\log\frac{1}{n}\sum_{i=0}^{n-1}
      R_\mu(X_0|X_{-i:-1})}-h_\mu
      \nonumber\\
    &=\mean D(P^{(\infty)}||\bar R^{(n)}),
  \end{align}
  where the last transition is due to equality
  $h_\mu=\mean \kwad{-\log P_\mu^{(\infty)}(X_0)}$ shown by Barron
  \cite[Proof of Theorem 1]{Barron85}.
\end{proof}

We recall the total variation distance of probability measures $P_1$
and $R_1$ on a measurable space $(\mathbb{X},\mathcal{X})$, defined as
\begin{align}
  \delta(P_1,R_1):=\sup_{A\in\mathcal{X}}\abs{P_1(A)-R_1(A)}.
\end{align}
If $P_1,R_1\ll\mu$ then
$\delta(P_1,R_1)=\frac{1}{2}\int \abs{P_\mu(x)-R_\mu(x)}d\mu(x)$.  We
also recall the Pinsker inequality \cite{CsiszarKorner11}, which reads
\begin{align}
  \label{Pinsker}
  \delta(P_1,R_1)\le
  \sqrt{\frac{D(P_1||R_1)}{2\log e}}.
\end{align}

In consequence, if a universal density exists and the entropy rate is
finite then the Ces\`aro mean density estimator is strongly consistent in the
total variation.
\begin{theorem}
  \label{theoOrnstein}
  Suppose that measure $R$ is universal in expectation with respect to
  a reference measure $\mu$. Then for every measure
  $P\in\mathbb{E}(\mu)$, we have
  \begin{align}
    \lim_{n\to\infty} \delta(P^{(\infty)},\bar R^{(n)})=
    \lim_{n\to\infty} \delta(P^{(n)},\bar R^{(n)})=
    0
    \text{ a.s.}
  \end{align}
\end{theorem}
\begin{proof}
  If $R$ is universal in expectation then for every measure
  $P\in\mathbb{E}(\mu)$, we have
  \begin{align}
    0&=\lim_{n\to\infty} \mean\kwad{-\log R_\mu(X_{1:n})}/n-h_\mu
       \ge \lim_{n\to\infty} \mean D(P^{(\infty)}||\bar R^{(n)})\ge 0.
  \end{align}
  Now by the Pinsker inequality (\ref{Pinsker}) and by the dominated
  convergence, we obtain
  \begin{align}
    0
    &=
    \lim_{n\to\infty} \mean\kwad{\delta(P^{(\infty)},\bar R^{(n)})}^2
      =
      \mean \kwad{\lim_{n\to\infty}\delta(P^{(\infty)},\bar R^{(n)})}^2,
  \end{align}
  which implies the almost sure convergence for
  $\delta(P^{(\infty)},\bar R^{(n)})$.  Since
  \begin{align}
    h_\mu=\lim_{n\to\infty}\mean \kwad{-\log P_\mu^{(n)}(X_0)}
  \end{align}
  for a stationary $P$ then we also have
  \begin{align}
    \lim_{n\to\infty} \mean D(P^{(n)}||\bar R^{(n)})=
    \lim_{n\to\infty} \mean D(P^{(\infty)}||\bar R^{(n)})=0.
  \end{align}
  Thus, analogously we derive the almost sure convergence for
  $\delta(P^{(n)},\bar R^{(n)})$.
\end{proof}

\subsection{Universal prediction}
\label{secPrediction}

Theorem \ref{theoOrnstein} strengthens and generalizes the celebrated
Ornstein theorem, originally stated for binary stationary ergodic
processes \cite{Ornstein78}.  Ornstein's theorem plays an important
role in the theory of universal prediction of binary processes with
the $0-1$ loss \cite{MorvaiWeiss21}. Analogously, we can apply Theorem
\ref{theoOrnstein} to develop a theory of universal prediction for
processes over an arbitrary countable alphabet, strengthening the
recent result of \cite{DebowskiSteifer22} by the way.

The exact development, recalling the basic facts from
\cite{DebowskiSteifer22}, is as follows.  For a countable alphabet
$\mathbb{X}$, we will consider densities with respect to the counting
measure $\gamma$, denoted without subscript $\gamma$ according to our
earlier convention.  A predictor is an arbitrary function
$f:\mathbb{X}^*\rightarrow\mathbb{X}$.  The predictor $f_P$ induced by
a probability measure $P$ is defined as a maximizer of conditional
probability,
\begin{align}
  \label{MetaPredictor}
  f_{P}(x_{1:n-1})\in\argmax_{x_{n}\in\mathbb{X}} P(x_{n}|x_{1:n-1}).
\end{align}  
In the problem of universal prediction we seek for a
measure-independent predictor that minimizes the relative frequency of
prediction mistakes. That is, we apply the $0-1$ loss rather than the
logarithmic loss encountered in the problem of universal coding.

By the Azuma-Hoeffding inequality \cite{Azuma67}, for any probability
measure $P$, the rate of mistakes is equal to the rate of their
conditional probabilities, namely,
\begin{align}
  \label{AzumaPredictor}
  \lim_{n\to\infty}\frac{1}{n}\sum_{i=1}^{n}
  \kwad{\boole{X_{i}\neq f(X_{1:i-1})}-P(X_{i}\neq f(X_{1:i-1})|X_{1:i-1})}=0
  \text{ a.s.},
\end{align}
see \cite[Theorem 3.5]{DebowskiSteifer22} for the derivation.
Therefore, as shown in \cite{Algoet94}, \cite[Theorem
3.5]{DebowskiSteifer22} using the Breiman ergodic theorem
\cite{Breiman57}, \cite[Theorem 12(c)]{Algoet94}, for any measure
$P\in\mathbb{E}$ and any predictor $f$, we have
\begin{align}
  \label{Predictor}
  \liminf_{n\to\infty}\frac{1}{n}\sum_{i=1}^{n}\boole{X_{i}\neq f(X_{1:i-1})}
  &\ge
    u
    \text{ a.s.},
\end{align}
where we define the unpredictability rate
\begin{align}
  u:=\mean\kwad{1-\max_{x\in\mathbb{X}} P(x|X_{-\infty:-1})}
  .
\end{align}
By contrast, for any measure $P\in\mathbb{E}$ and its induced
predictor $f_P$, we have
\begin{align}
  \lim_{n\to\infty}\frac{1}{n}\sum_{i=1}^{n}\boole{X_{i}\neq f_P(X_{1:i-1})}
  &=
    u
    \text{ a.s.}
\end{align}

Thus by an analogy to universal measures, we propose universal
predictors.
\begin{definition}
  \label{defiPUniAS}
  A predictor $f$ is called strongly universal with respect to a
  reference measure $\mu$ if for any measure $P\in\mathbb{E}(\mu)$,
  \begin{align}
    \label{PUniAS}
    \lim_{n\to\infty}\frac{1}{n}\sum_{i=1}^{n}\boole{X_{i}\neq f(X_{1:i-1})}
    &=
      u
    \text{ a.s.}
  \end{align}
\end{definition}
\begin{definition}
  \label{defiPUniExp}
  A predictor $f$ is called strongly universal with respect to a
  reference measure $\mu$ if for any measure $P\in\mathbb{E}(\mu)$,
  \begin{align}
    \label{PUniExp}
    \lim_{n\to\infty}\mean
    \frac{1}{n}\sum_{i=1}^{n}\boole{X_{i}\neq f(X_{1:i-1})}
    &=
      u.
  \end{align}
\end{definition}
\begin{definition}
  \label{defiPUniIP}
  A predictor $f$ is called strongly universal with respect to a
  reference measure $\mu$ if for any measure $P\in\mathbb{E}(\mu)$,
  \begin{align}
    \label{PUniIP}
    \lim_{n\to\infty}\frac{1}{n}\sum_{i=1}^{n}\boole{X_{i}\neq f(X_{1:i-1})}
    &=
      u
    \text{ i.p.}
  \end{align}
\end{definition}
\noindent
\emph{Remark:} By analogy, a predictor that satisfies (\ref{PUniAS}),
(\ref{PUniExp}), and (\ref{PUniIP}) is simply called universal with
respect to $\mu$. But we have implications (\ref{PUniAS}) $\implies$
(\ref{PUniExp}) $\implies$ (\ref{PUniIP}) always. Implication
(\ref{PUniAS}) $\implies$ (\ref{PUniExp}) follows by the dominated
convergence.  Implication (\ref{PUniExp}) $\implies$ (\ref{PUniIP})
follows by the Markov inequality and the almost sure lower bound
(\ref{Predictor}).  The exact proofs resemble proofs of Propositions
\ref{theoUniExp} and \ref{theoUniIP}. Thus each strongly universal
predictor is universal.

In work \cite[Theorems 3.12 and 3.18]{DebowskiSteifer22}, it was shown
that for a finite alphabet $\mathbb{X}$ and a strongly universal
measure $R$, the induced predictor $f_R$ is universal if we have a
uniform bound for conditional probabilities of form
\begin{align}
    \label{CondDomination}
    -\log R(x_{n+1}|x_1^n)&\le\epsilon_n\sqrt{n/\log n},
                          \quad \lim_{n\to\infty} \epsilon_n=0.
\end{align}
Being uniform in symbols, condition (\ref{CondDomination}) can be
satisfied only if the alphabet is finite.  Subsequently, we will show
that condition (\ref{CondDomination}) can be dropped if measure $R$ is
universal in expectation and if we consider predictor $f_{\bar R}$,
induced by the Ces\`aro mean measure $\bar R$, rather than predictor
$f_R$, induced by the original measure $R$. Hence we have a universal
predictor also for a countably infinite alphabet.
\begin{theorem}
  Consider a countable alphabet $\mathbb{X}$.  Suppose that measure
  $R$ is universal in expectation with respect to a reference measure
  $\mu$. The Ces\`aro mean predictor $f_{\bar R}$ is universal with
  respect to measure $\mu$.
\end{theorem}
\begin{proof}
  Consider a measure $P\in\mathbb{E}(\mu)$.  By Theorem
  \ref{theoOrnstein}, we have a generalization of the Ornstein theorem
  \cite{Ornstein78}, namely,
  \begin{align}
    \label{OrnsteinCountable}
    \lim_{n\to\infty} 2\delta(P^{(n)},\bar R^{(n)})=
    \lim_{n\to\infty} \sum_{x\in\mathbb{X}}
    \abs{P(x|X_{-n:-1})-\bar R(x|X_{-n:-1})}=0
    \text{ a.s.}
  \end{align}
  since $P^{(n)},\bar R^{(n)}\ll\gamma$.  As it follows from a simple
  application of the Breiman ergodic theorem \cite{Breiman57},
  \cite[Theorem 12(c)]{Algoet94}, statement (\ref{OrnsteinCountable})
  implies a generalization of the Bailey theorem \cite{Bailey76},
  namely,
  \begin{align}
    \lim_{n\to\infty} \frac{1}{n}\sum_{i=1}^n
    \sum_{x\in\mathbb{X}}
    \abs{P(x|X_{1:i-1})-\bar R(x|X_{1:i-1})}=0
    \text{ a.s.}
  \end{align}
  Since we have a so called prediction inequality
  \begin{align}
    &P(X_{i}=f_P(X_{1:i-1})|X_{1:i-1})-
      P(X_{i}=f_{\bar R}(X_{1:i-1})|X_{1:i-1})
      \nonumber\\
    &\qquad\le
    \sum_{x\in\mathbb{X}}
    \abs{P(x|X_{1:i-1})-\bar R(x|X_{1:i-1})}
  \end{align}
  noticed in work \cite[Proposition 3.11]{DebowskiSteifer22} then by
  corollary (\ref{AzumaPredictor}) of the Azuma-Hoeffding inequality,
  we obtain
  \begin{align}
    \lim_{n\to\infty}\frac{1}{n}\sum_{i=1}^{n}
    \kwad{\boole{X_{i}\neq f_{\bar R}(X_{1:i-1})}-\boole{X_{i}\neq f_P(X_{1:i-1})}}=0
    \text{ a.s.}
  \end{align}
  That is, the Ces\`aro mean predictor $f_{\bar R}$ is universal.
\end{proof}

The Ces\`aro mean predictor $f_{\bar R}$ for $R=\NPD$ that applies the
PPM measures is quite complex. It contains five nested maximizations,
summations, and products that may contribute to a pessimistic time
complexity $O(n^5)$, where $n$ is the length of the sample. In the
future research, it would be advisable to seek for a universal
predictor with a smaller time complexity.

\section{Examples}
\label{secExamples}

As some supplementary examples, in this section, we scale down the NPD
measure to two cases where consistent estimation of the entropy rate
is not possible in general. These are the countably infinite alphabet
and the real line. The general infeasibility of consistent estimation
in these cases becomes intuitive by virtue of additional assumptions
that we are bound to make. Once these conditions are met, the
corrected NPD estimator is strongly consistent.

\subsection{Countably infinite alphabet}
\label{secInfinite}

Subsequently, let us consider a countably infinite alphabet
$\mathbb{X}$. We recall that we have earlier adopted notation
$P_1:\mathbb{X}\ni x\mapsto P(x)$ for the marginal density of measure
$P$ on $(\mathbb{X}^{\mathbb{Z}},\mathcal{X}^{\mathbb{Z}})$ with
respect to the counting measure $\gamma(A)=\card A$ on
$(\mathbb{X},\mathcal{X})$. Thus for probability measures $P$ and $R$,
we denote the Shannon entropy and the Kullback-Leibler divergence
taken with respect to their marginal densities:
\begin{align}
  H(P_1)&:=
  -\sum_{x\in\mathbb{X}} P(x)\log P(x),
  \\
  D(P_1||R_1)&:=
  \sum_{x\in\mathbb{X}} P(x)\log\frac{P(x)}{R(x)}.
\end{align}
It may appear a bit surprising that we have $P_1$ on the left hand
side and $P$ on the right hand side but we prefer not to multiply the
notational conventions, which are sufficiently overloaded to deal with
quantization levels.

We recall that there is no consistent conditional density estimator
for a countably infinite alphabet, in general.
\begin{theorem}[\mbox{\cite[Theorem 2]{GyorfiPaliMeulen94}}]
  \label{theoCounterexample}
  Let the alphabet $\mathbb{X}$ be countably infinite.  Let
  $(R^{(n)})_{n\ge 1}$ be an arbitrary conditional density
  estimator. Then there is a memoryless source $P$ such that
  $H(P_1)<\infty$ and $D(P_1||R^{(n)})=\infty$ a.s.\ for all $n$.
\end{theorem}

In view of Theorems \ref{theoInduced} and \ref{theoCounterexample},
for a countably infinite alphabet $\mathbb{X}$, for each measure $R$
there exists a memoryless source $P$ with $H(P_1)<\infty$ such that
\begin{align}
  \lim_{n\to\infty} \mean\kwad{-\log R(X_{1:n})}/n&=\infty.
\end{align}
Since each memoryless source $P$ is stationary ergodic with entropy
rate $H(P_1)$ we obtain the known result.
\begin{theorem}[\mbox{\cite{Kieffer78}, \cite[Theorem
    3]{GyorfiPaliMeulen94}}]
  There is no measure universal in expectation with respect to a
  countably infinite alphabet.
\end{theorem}

For a countably infinite alphabet $\mathbb{X}$, consider now a
reference measure $\mu$ to be contrasted with the counting measure
$\gamma$. We suppose that $\mu\ll \gamma$ and $\mu\gg \gamma$.
Entropies of $P$ with respect to $\mu$ are written as $h_\mu(n)$ or
$h_\mu$, whereas entropies of $P$ with respect to $\gamma$ are written
as $h(n)$ or $h$. We have
\begin{align}
  \label{PPmu}
  P(x_{1:n})=
  P_\mu(x_{1:n})\frac{d\mu^n}{d\gamma^n}(x_{1:n})=
  P_\mu(x_{1:n})\prod_{i=1}^n\mu(x_i).
\end{align}

Let us write the marginal cross entropy
\begin{align}
 H(P_1||R_1):=H(P_1)+D(P_1||R_1)=-\sum_{x\in\mathbb{X}} P(x)\log R(x).
\end{align}
If $H(P_1||\mu)<\infty$ then from equation (\ref{PPmu}), applying
stationarity, we obtain
\begin{align}
  h(n)
  &=h_\mu(n)
    -n\sum_{x\in\mathbb{X}} P(x)\log \mu(x)
    =h_\mu(n)+n H(P_1||\mu)
    .
\end{align}
Hence $h=h_\mu+H(P_1||\mu)\ge 0$. If $\mu$ is a probability measure
then $h_\mu\le 0$. Consequently, we have a sufficient condition
\begin{align}
  P\in\mathbb{E},
  \;
  \mu(x)>0 \text{ for all } x\in\mathbb{X},
  \; 
  H(P_1||\mu)<\infty
  \implies
  P\in\mathbb{E}(\mu).
\end{align}

In particular, we may estimate the entropy rate in the following way,
compare it with more complicated characterizations in
\cite{SilvaPiantanida20}:
\begin{theorem}
  \label{theoInfiniteNPD}
  Consider a countably infinite alphabet $\mathbb{X}$ and a
  probability measure $\mu$ such that $\mu(x)>0$ for all
  $x\in\mathbb{X}$. Let $P\in\mathbb{E}$ with
  $H(P_1||\mu)<\infty$. Then
  \begin{align}
    \label{InfiniteNPD}
    \lim_{n\to\infty}\frac{1}{n}\kwad{-\log\NPD_\mu(X_{1:n})-\sum_{i=1}^n
    \log\mu(X_i)}=h
    \text{ a.s.}
  \end{align}
\end{theorem}
\noindent
\emph{Remark:} In particular, estimator (\ref{InfiniteNPD}) is
strongly consistent for any source over a finite but unknown alphabet.
\begin{proof}
  Quantity $-n^{-1}\log\NPD_\mu(X_{1:n})$ is a strongly consistent
  estimator of $h_\mu$, whereas $-n^{-1}\sum_{i=1}^n \log\mu(X_i)$ is
  a strongly consistent estimator of $H(P_1||\mu)$ by the Birkhoff
  ergodic theorem. Since $H(P_1||\mu)<\infty$ then the sum of these two
  estimators converges to $h$.
\end{proof}

Let us consider a particular quantization of the set of natural
numbers.
\begin{example}[incremental partitions]
  \label{exNaturalNumbers}
  Let the alphabet be the set of natural numbers,
  $\mathbb{X}=\mathbb{N}$. Suppose that $\mu(m)>0$ for all
  $m\in\mathbb{N}$.  Let $P\in\mathbb{E}$ with cross entropy
  $H(P_1||\mu)<\infty$.  Let the universal measures in (\ref{NPDml})
  be $R^l=\PPM^{\card \chi_l}$.  In this case, estimator
  (\ref{InfiniteNPD}) differs from the PPM measure since
  $\mu(X_i)\le \mu(X_i^l)$. To observe simplifications of estimator
  (\ref{InfiniteNPD}), let us take partitions
\begin{align}
  \chi_l:=\klam{\klam{1},\klam{2},...,\klam{l},\klam{l+1,l+2,...}}
  ,
  \quad
  l\ge 0
  .
\end{align}
Let us denote the optimal quantization level $Q_n$ and the optimal
Markov order $R_n$ for sample $X_{1:n}$ as some elements
\begin{align}
  (Q_n,R_n)\in\argmax_{q,r\ge 0}
  \frac{\PPM^{q+1}_r(X_{1:n}^{q})}{\mu^n(X_{1:n}^{q})}
  .
\end{align}
Since we have $\PPM^{q+1}_r(x_{1:n})<\PPM^{q}_r(x_{1:n})$ in general
and $\PPM^{q}_r(x_{1:n})=q^{-n}$ for $r\ge n-1$, we may fix pair
$(Q_n,R_n)$ so that
\begin{align}
  Q_n&\le \max X_{1:n},
  &
    R_n&\le n-1.
\end{align}
Using the well known idea for the PPM measure, we may also bound
  \begin{align}
    \label{NPDBound}
    0\le
    -\log\NPD_\mu(X_{1:n})
    +\log\frac{\PPM^{Q_n+1}_{R_n}(X_{1:n}^{Q_n})}{\mu^n(X_{1:n}^{Q_n})}
    \le
    -\log w_{Q_n}-\log w_{R_n}.
  \end{align}
\end{example}

Thus we may specialize Theorem \ref{theoInfiniteNPD} as the following
proposition.
\begin{proposition}
  Consider the setting of Example \ref{exNaturalNumbers}. Suppose that
  \begin{align}
  \lim_{n\to\infty}\frac{\log \max X_{1:n}}{n}=0 \text{ a.s.}
  \end{align}
  Then we have 
  \begin{align}
    \label{InfinitePPM}
    \lim_{n\to\infty}\frac{1}{n}
    \kwad{-\log\PPM^{Q_n+1}_{R_n}(X_{1:n}^{Q_n})+C_n(Q_n)}=h
    \text{ a.s.},
  \end{align}
  where we define
  \begin{align}
    C_n(q)
    &:=
      -\sum_{i=1}^n \log\frac{\mu(X_i)}{\mu(X_i^q)}
      =
      -\sum_{i=1}^n \boole{X_i> q}\log\mu(X_i|X_i> q)\ge 0.
  \end{align}
\end{proposition}
\begin{proof}
  Convergence (\ref{InfinitePPM}) follows from convergence
  (\ref{InfiniteNPD}) and the sandwich bound (\ref{NPDBound}).
\end{proof}

Term $C_n(Q_n)$ is a non-negative correction of the finite-alphabet
entropy estimator taken for the optimal quantization level and the
optimal Markov order.  In the following, we will show the this
correction vanishes ultimately almost surely if the process is upper
bounded.  To shed some light on this issue, let us introduce the
minimal sufficient quantization level of a probability measure
$P\in\mathbb{E}(\mu)$ relative to measure $\mu$, which is defined as
\begin{align}
  \label{QuantizationLevel}
  Q:=\inf\klam{q\ge 0:
  \frac{P(X_0^{q}|X_{-\infty:-1})}{\mu(X_0^{q})}
  =
  \frac{P(X_0|X_{-\infty:-1})}{\mu(X_0)}
  \text{ a.s.}
  }.
\end{align}
In particular, the minimal sufficient quantization level is $Q=0$ for
the memoryless source $P=\mu^{\mathbb{Z}}$, which seems somewhat
counterintuitive. What is more intuitive, we have $Q\le M$ if
$X_i\le M$ since $X_i^M=\klam{X_i}$ holds almost surely.

As an auxiliary result, we will show that for any stationary ergodic
measure $P$ with entropy $h_\mu<\infty$, statistic $Q_n$ does not
underestimate parameter $Q$.
\begin{proposition}
  \label{theoNoUnderestimation}
  Consider the setting of Example \ref{exNaturalNumbers}. We have
  \begin{align}
    \liminf_{n\to\infty} Q_n\ge Q \text{ a.s.}
  \end{align}
\end{proposition}
\begin{proof}
  We apply the proof idea of \cite[Theorem 6.12]{Debowski21} for the
  inconsistent estimator of the Markov order given by the PPM measure,
  the inconsistency result due to Csiszar and Shields
  \cite{CsiszarShields00}. By contradiction, let us assume that
  $\liminf_{n\to\infty} Q_n=q$ holds with a positive probability for
  some $q<Q$. Then on the respective random points, we obtain
\begin{align}
  h_\mu
  &=
    \lim_{n\to\infty} \frac{1}{n}\kwad{-\log\NPD_\mu(X_{1:n})}
    \text{ (by universality of NPD)}
    \nonumber\\
  &=
    \lim_{n\to\infty}
    \frac{1}{n}
    \kwad{-\log\frac{\PPM^{Q_n+1}_{R_n}(X_{1:n}^{Q_n})}{\mu^n(X_{1:n}^{Q_n})}}
    \text{ (by (\ref{NPDBound}))}
    \nonumber\\
  &\ge
    \liminf_{n\to\infty}
    \frac{1}{n}
    \kwad{-\log\frac{\PPM^{q+1}_{R_n}(X_{1:n}^{q})}{\mu^n(X_{1:n}^{q})}}
    \text{ (since $Q_n=q$ infinitely often)}
    \nonumber\\
  &\ge
    \lim_{n\to\infty}
    \frac{1}{n}\kwad{-\log\frac{P(X_{1:n}^{q})}{\mu^n(X_{1:n}^{q})}}
    \text{ (by the Barron inequality)}
    \nonumber\\
  &=
    \mean
    \kwad{-\log\frac{P(X_0^{q}|X_{-\infty:-1}^{q})}{\mu(X_0^{q})}}
    \text{ (by the Breiman ergodic theorem)}
    \nonumber\\
  &\ge
    \mean
    \kwad{-\log\frac{P(X_0^{q}|X_{-\infty:-1})}{\mu(X_0^{q})}}
    \text{ (conditioning decreases entropy)}
    \nonumber\\
  &>
    \mean
    \kwad{-\log\frac{P(X_0|X_{-\infty:-1})}{\mu(X_0)}}
    \text{ (by $q<Q$ and the definition of $Q$)}
    \nonumber\\
  &=
    h-H(P_1||\mu)
    .
\end{align}
Inequality $h_\mu>h-H(P_1||\mu)$ cannot be true so our assumption is
false. Thus $\liminf_{n\to\infty} Q_n\ge Q$ almost surely.
\end{proof}

Now we prove the desired claim about vanishing of $C_n(Q_n)$.
\begin{proposition}
  Consider the setting of Example \ref{exNaturalNumbers} with
  $X_i\le M<\infty$. Then
  \begin{align}
    \lim_{n\to\infty} C_n(Q_n)=0 \text{ a.s.}
  \end{align}
\end{proposition}
\begin{proof}
  Without loss of generality, suppose that $M:=\esssup X_i<\infty$
  almost surely. We have $Q\le M$.  Observe that
  \begin{align}
    Q=\inf\klam{q\ge 0:
    \forall_{m>q}
    P(X_0=m|X_{-\infty:-1})=K\mu(m)
    \text{ a.s.}
    },
  \end{align}
  where $K\ge 0$ is a certain random variable.  Since $\mu(m)>0$ and
  $P(X_0=m|X_{-\infty:-1})=0$ for $m>M$, hence $K=0$ almost surely and
  $P(X_0=m|X_{-\infty:-1})=0$ for all $m>Q$. In other words,
  $Q\ge\esssup X_i=M\ge \max X_{1:n}\ge Q_n$. Consequently,
  $\lim_{n\to\infty} Q_n=M$ holds almost surely in view of Proposition
  \ref{theoNoUnderestimation}. Since also $X_i^M=\klam{X_i}$ almost
  surely, we have $C_n(M)=0$ and $\lim_{n\to\infty} C_n(Q_n)=0$ almost
  surely in this case.
\end{proof}

\subsection{Real line and Gaussian processes}
\label{secReal}

Consider alphabet $\mathbb{X}=\mathbb{R}$ and the reference measure
$\mu$ being the normal distribution $N(m,\sigma^2)$ to be contrasted
with the Lebesgue measure $\lambda$. Entropies of $P$ with respect to
$\mu$ are written as $h_\mu(n)$ or $h_\mu$, whereas entropies of $P$
with respect to $\lambda$ are written as $h_\lambda(n)$ or
$h_\lambda$.  We have
\begin{align}
  \frac{d\mu}{d\lambda}(x)=\frac{1}{\sigma\sqrt{2\pi}}
  \exp\okra{-\frac{(x-m)^2}{2\sigma^2}}.
\end{align}
Let us write the rescaled second moment of the marginal distribution
\begin{align}
  M_2:=\int_{-\infty}^\infty \frac{(x-m)^2}{2\sigma^2}
  P_\lambda(x)dx
  = \frac{\var X_i+(\mean X_i-m)^2}{2\sigma^2},
\end{align}
denoting variance $\var X_i:=\mean(X_i-\mean X_i)^2$.

If $M_2<\infty$, we obtain like in Section \ref{secInfinite} that
\begin{align}
  h_\lambda(n)
  &=h_\mu(n)
    -n\int_{-\infty}^\infty P_\lambda(x)\log \frac{d\mu}{d\lambda}(x) dx.
    \nonumber\\
  &=h_\mu(n)+n\kwad{M_2\log e+\log\sigma\sqrt{2\pi}}.
\end{align}
Hence $h_\lambda=h_\mu+M_2\log e+\log\sigma\sqrt{2\pi}$.  Since $\mu$
is a probability measure then $h_\mu\le 0$. Consequently, we have a
sufficient condition
\begin{align}
  P\in\mathbb{E},
  \;
  P_n\ll\lambda^n,
  \;
  \abs{\mean X_i}<\infty,
  \;
  \var X_i<\infty,
  \;
  \abs{h_\lambda}<\infty
  \implies
  P\in\mathbb{E}(\mu).
\end{align}

In this case, we may estimate the entropy rate in the following way.
\begin{theorem}
  \label{theoGaussNPD}
  Consider $\mathbb{X}=\mathbb{R}$. Let $\mu$ be the normal
  distribution $N(m,\sigma^2)$, whereas $\lambda$ be the Lebesgue
  measure.  Suppose that $P\in\mathbb{E}$ with $P_n\ll\lambda^n$,
  $\abs{\mean X_i}<\infty$, $\var X_i<\infty$, and
  $\abs{h_\lambda}<\infty$. Then
  \begin{align}
    \label{GaussNPD}
    \lim_{n\to\infty}\frac{1}{n}\kwad{-\log\NPD_\mu(X_{1:n})+
    \kwad{\sum_{i=1}^n
    \frac{(X_i-m)^2}{2\sigma^2}}\log e}+
    \log\sigma\sqrt{2\pi}=h_\lambda
    \text{ a.s.}
  \end{align}
\end{theorem}
\begin{proof}
  Quantity $-n^{-1}\log\NPD_\mu(X_{1:n})$ is a strongly consistent
  estimator of the entropy rate $h_\mu$, whereas
  $-n^{-1}\sum_{i=1}^n (X_i-m)^2/2\sigma^2$ is a strongly consistent
  estimator of the rescaled second moment $M_2$. Since
  $M_2<\infty$ then the linear combination of these
  estimators tends to $h_\lambda$.
\end{proof}
\noindent
Thus, the corrected NPD estimator (\ref{GaussNPD}) is a strongly
consistent estimator of the entropy rate for all non-deterministic
stationary ergodic Gaussian processes since $\abs{\mean X_i}<\infty$,
$\var X_i<\infty$, and $\abs{h_\lambda}<\infty$ holds in this case.
As we have mentioned in Section \ref{secIntroduction}, Feutrill and
Roughan \cite{FeutrillRoughan21} supposed that the NPD estimator can
estimate the entropy rate of a Gaussian process but they did not carry
out this idea rigorously enough.

Consider the dyadic filtration (\ref{Dyadic}) from Example
\ref{exQuantiles}.  In this case, computing quantizations $X_i^l$
requires computing quantiles of the normal distribution. An
interesting question is whether such a filtration is
optimal. Moreover, we may suppose that the corrected NPD estimator
(\ref{GaussNPD}) improves if we take parameters $m$ and $\sigma$ close
to the expectation and variance of $X_i$ with respect to
$P$. Obviously, we can specialize Theorem \ref{theoGaussNPD} as the
following proposition.
\begin{proposition}
  Consider $\mathbb{X}=\mathbb{R}$.  Let $\mu$ be the normal
  distribution $N(m,\sigma^2)$, whereas $\lambda$ be the Lebesgue
  measure. Suppose that $P\in\mathbb{E}$ with $P_n\ll\lambda^n$,
  $\mean X_i=m$, $\var X_i=\sigma^2$, and
  $\abs{h_\lambda}<\infty$. Then
  \begin{align}
    \lim_{n\to\infty}\frac{1}{n}\kwad{-\log\NPD_\mu(X_{1:n})}+
    \log\sigma\sqrt{2\pi e}=h_\lambda
    \text{ a.s.}
  \end{align}
\end{proposition}

The problem with the above estimator is that we
need to know the exact expectation and the variance of $X_i$. Can we
estimate them from sample $X_{1:n}$ and plug the result into the NPD
estimator? Consider random measures $\mu_n\sim N(m_n,\sigma^2_n)$,
where
\begin{align}
  m_n&:=\frac{1}{n} \sum_{i=1}^n X_i,
  &
  \sigma_n^2&:=\frac{1}{n} \sum_{i=1}^n (X_i-m_n)^2.
\end{align}
We may ask whether there holds still convergence
\begin{align}
  \lim_{n\to\infty}\frac{1}{n}\kwad{-\log\NPD_{\mu_n}(X_{1:n})}+
    \log\sigma_n\sqrt{2\pi e}=h_\lambda
    \text{ a.s.}
\end{align}
However, this question is not stated precisely enough. Let us note
that measure $\NPD$ depends implicitly on the quantization path
$\mathcal{X}_l\uparrow \mathcal{X}$ and measure $\mu$. When we variate
the reference measure $\mu\to\mu_n$, it is not clear whether we have
also to variate $\sigma$-fields $\mathcal{X}_l\to
\mathcal{X}_{l,n}$. In any case, the proof technique of Theorem
\ref{theoFiniteAS} works no longer and we have no simple guarantee of
consistency.

\section{Conclusion}
\label{secConclusion}

Drawing an inspiration from the non-parametric differential (NPD)
entropy rate estimator by Feutrill and Roughan
\cite{FeutrillRoughan21}, we have constructed a universal NPD density
that works for stationary ergodic sources on any countably generated
measurable space---as long as we agree to define the entropy relative
to a finite reference measure. Our theoretical development of the NPD
estimator is analogical to Ryabko's \cite{Ryabko88en2,Ryabko08}
development of the PPM measure by Cleary and Witten
\cite{ClearyWitten84}. Whereas Ryabko considered a countably infinite
mixture of source estimates of distinct Markov orders, we have
considered an analogous mixture of source estimates of distinct
quantization levels.

As we have shown, the NPD density solves the problem of consistent
entropy rate estimation. Moreover, it can be used to obtain a strongly
consistent Ces\`aro mean estimator of the conditional density given an
infinite past, in the total variation, using the idea of
\cite{GyorfiPaliMeulen94}. This in turn solves the problem of
universal prediction with the $0-1$ loss for a countable alphabet,
cf.\ \cite{DebowskiSteifer22}.  The NPD density can also shed light on
sufficient conditions for consistent estimation of the entropy rate
with respect to infinite reference measures, cf.\
\cite{SilvaPiantanida20}.

There is one hanging gun that has not shot in this play,
however. Among the applications of universal coding, in Section
\ref{secIntroduction}, we have mentioned estimation of the (hidden)
Markov order, see \cite{MerhavGutmanZiv89,ZivMerhav92} for classical
references. Analogously, we may ask whether a similar approach can be
developed for estimation of the minimal sufficient quantization level
for real-valued stochastic processes. The general method of Markov
order estimation requires comparing universal measures or codes with
the maximum likelihood. However, we are not sure whether this method
can be translated to quantization level estimation and whether the
concept of the minimal sufficient quantization level of an arbitrary
process can be reasonably defined, see definition
(\ref{QuantizationLevel}) which is somewhat counterintuitive. Solving
this issue is deferred to another work. We recall that we have also
stated some related open questions in Section \ref{secReal} that
concern the optimal quantization of Gaussian processes.

Another important open topic is the computational complexity of
universal densities.  The infinite series (\ref{NPDm}) can be
truncated or approximated in some cases, as we have shown in Examples
\ref{exQuantiles} and \ref{exNaturalNumbers}. However, we have not
cared about the time complexity or the speed of convergence of
universal densities, being satisfied by their general computability or
consistency. In the future research, one should develop explicit
bounds for time and memory complexity of the total NPD density and to
find other universal densities that can be computed faster. Yet
another idea for future research is to investigate the rate of
convergence of the NPD entropy rate estimates and of the Ces\`aro mean
estimator of conditional density given an infinite past. We think that
both mathematical statistics and information theory can benefit from
such analyses.

\section*{Funding}

This work was supported by the National Science Centre Poland
grant 2018/31/B/HS1/04018.

\section*{Acknowledgments}

I thank two anonymous referees who encouraged me to work on improving
the structure of this paper.

\bibliographystyle{IEEEtran}

\bibliography{0-journals-abbrv,0-publishers-abbrv,ai,ql,math,mine,tcs,books}

\end{document}